\documentclass{llncs}

\usepackage{amsmath,amssymb}
\usepackage[linesnumbered]{algorithm2e}
\usepackage{graphicx}
\usepackage{enumerate}

\DeclareMathOperator{\scat}{{\rm sc}}

\DeclareMathOperator{\merg}{{\rm merge}}

\newcommand{\cP}{{\mathcal P}}
\newcommand{\cQ}{{\mathcal Q}}
\renewcommand{\P}{{\mathsf P}}
\newcommand{\NP}{{\mathsf{NP}}}
\newcommand{\coNP}{{\mathsf{coNP}}}
\newcommand\pcase[1] {\medskip\par\noindent \textbf{#1}}

\begin{document}
\title{Linear-Time Algorithms for Scattering Number and Hamilton-Connectivity of Interval Graphs
\thanks{Paper supported by Royal Society Joint Project Grant JP090172.}}

\author
{Hajo Broersma\inst{1} 
\and 
Ji\v{r}\'{\i} Fiala\inst{2}\thanks{Author supported by the GraDR-EuroGIGA project GIG/11/E023 and by the project Kontakt LH12095.} 
\and
Petr A. Golovach\inst{3}  
\and 
Tom\'a\v{s} Kaiser\inst{4}\thanks{Author supported by project P202/11/0196 of the Czech Science Foundation.}  
\and 
Dani\"el Paulusma\inst{5}\thanks{Author supported by EPSRC (EP/G043434/1).}  
\and Andrzej Proskurowski\inst{6}
}

\institute{
Faculty of EEMCS, University of Twente, Enschede, {\tt h.j.broersma@utwente.nl} 
\and
Department of Applied Mathematics, Charles University, Prague, {\tt fiala@kam.mff.cuni.cz}
\and
Institute of Computer Science, University of Bergen, {\tt petr.golovach@ii.uib.no}
\and
Department of Mathematics, University of West Bohemia, Plze\v{n}, {\tt kaisert@kma.zcu.cz}
\and
School of Engineering and Computing Sciences, Durham University, {\tt daniel.paulusma@durham.ac.uk}
\and
Department of Computer Science, University of Oregon, Eugene, {\tt andrzej@cs.uoregon.edu}
}

\maketitle

\begin{abstract}\noindent
Hung and Chang showed that for all $k\geq 1$ an interval graph has a path cover of size at most $k$ if and only if its scattering number is at 
most $k$. They also showed that an interval graph has a Hamilton cycle if and only if its scattering number is at most $0$.
We complete this characterization by proving that for all $k\leq -1$ an interval graph is $-(k+1)$-Hamilton-connected if and only if its scattering number is at most $k$. 
We also give an $O(m+n)$ time algorithm for computing the scattering number of an interval graph with $n$ vertices an $m$ edges, which improves the $O(n^4)$ time bound of Kratsch, Kloks and M\"uller.  
As a consequence of our two results the maximum $k$ for  which an interval graph is $k$-Hamilton-connected can be computed in $O(m+n)$ time. 
\end{abstract}

\section{Introduction}\label{s:intro}

The {\sc Hamilton Cycle} problem is that of 
testing whether a given graph has a Hamilton cycle, i.e., a cycle passing through all the vertices.
This problem
 is one of the most notorious $\NP$-complete problems within Theoretical Computer Science.
It remains $\NP$-complete on many graph classes such as the classes of planar cubic 3-connected graphs~\cite{GJT76},
chordal bipartite graphs~\cite{Mu96},  and strongly chordal split graphs~\cite{Mu96}. In contrast, for interval graphs, Keil~\cite{Ke85} showed in 1985 that {\sc Hamilton Cycle} can be solved in $O(m+n)$ time, thereby
strengthening an earlier result of Bertossi~\cite{Be83} for proper interval graphs.
Bertossi and Bonucelli~\cite{BB86} proved that {\sc Hamilton Cycle} is $\NP$-complete for undirected path graphs, double interval graphs and rectangle graphs, all three of which are classes of intersection graphs that contain the class of interval graphs.
We examine whether the linear-time result of Keil~\cite{Ke85}  can be strengthened on interval graphs to hold for other connectivity properties, which are $\NP$-complete to verify in general.
This line of research is well embedded in the literature. 
Before surveying existing work and presenting our new results, we first give the necessary terminology.

\subsection{Terminology}\label{s-term}

We only consider undirected finite graphs with no self-loops and no multiple edges. We refer to the textbook of Bondy and Murty~\cite{BM08} for any undefined graph terminology.
Throughout the paper we let $n$ and $m$ denote the number of vertices and edges, respectively, of the input graph.

Let $G=(V,E)$ be a graph. If $G$ has a {\em Hamilton cycle\/}, i.e., a cycle containing all the vertices of $G$, then $G$ is  {\em hamiltonian}.
Recall that the corresponding $\NP$-complete decision problem is called {\sc Hamilton Cycle}.
If $G$ contains a {\em Hamilton path}, i.e., a path containing all the vertices of $G$, then $G$ is  {\em traceable}. In this case, the corresponding decision problem is called the {\sc Hamilton Path} problem,
which is also well known to be $\NP$-complete (cf.~\cite{GJ79}).
The problems {\sc 1-Hamilton Path} and {\sc 2-Hamilton Path} are those of testing whether a given graph has a Hamilton path that starts in some given vertex $u$ or that is between two given vertices $u$ and $v$, respectively. Both problems are $\NP$-complete by a straightforward reduction from {\sc Hamilton Path}. The {\sc Longest Path} problem is to compute the maximum length
of a path in a given graph. This problem is $\NP$-hard by a reduction from {\sc Hamilton Path} as well.

Let $G=(V,E)$ be a graph.
If for each two distinct vertices $s,t\in V$ there exists a Hamilton path with end-vertices $s$ and $t$, then
$G$ is {\em Hamilton-connected\/}. If $G-S$ is Hamilton-connected for every set $S\subset V$ with $|S|\le k$ for some integer $k\geq 0$, then
$G$ is {\em $k$-Hamilton-connected\/}. Note that a graph is Hamilton-connected if and only if it 0-Hamilton-connected. The  {\sc Hamilton Connectivity} problem is that of computing the
maximum value of $k$ for which a given graph is $k$-Hamilton-connected.  Dean~\cite{De93} showed that already deciding whether $k=0$ is $\NP$-complete.
Ku\v{z}el,  Ryj{\'a}\v{c}ek and Vr{\'a}na~\cite{KRV12} proved this for $k=1$. 
A straightforward generalization of the latter result yields the same for any integer $k\geq 1$.
As an aside, the {\sc Hamilton Connectivity} problem has recently been studied by Ku\v{z}el, Ryj\'a\v{c}ek and Vr\'ana~\cite{KRV12}, who showed that $\NP$-completeness of the case $k=1$
for line graphs would disprove the conjecture of Thomassen that every 4-connected line graph is hamiltonian, unless $\P=\NP$. 

A {\it path cover} of a graph $G$ is a set of mutually vertex-disjoint paths
 $P_1,\ldots,P_k$ with  $V(P_1)\cup \cdots \cup V(P_k)=V(G)$. The size of a smallest path cover is denoted by $\pi(G)$.
The {\sc Path Cover} problem is to compute this number, whereas the 
{\sc 1-Path Cover} problem is to compute the size of a smallest path cover that contains a path in which some given vertex $u$ is an end-vertex.
Because a Hamilton path of a graph is a path cover of size 1, 
{\sc Path Cover} and {\sc 1-Path Cover} are $\NP$-hard via a reduction from {\sc Hamilton Path} and {\sc 1-Hamilton Path}, respectively.

We denote the number of connected components of a graph $G=(V,E)$ by $c(G)$. 
A subset $S\subset V$ is a
{\em vertex cut\/} of $G$ if $c(G-S)\geq 2$, and $G$ is called {\it $k$-connected} if the size of a smallest vertex cut of $G$ is at least $k$.
We say that $G$ is \mbox{{\em{$t$}-tough}} if $\vert S \vert \ge t \cdot 
c(G-S)$ for every vertex cut $S$ of  $G$. The {\em toughness\/} $\tau(G)$ of a graph $G=(V,E)$ was defined by Chv\'atal~\cite{Ch73} as
\[\tau(G)= \min\big{\{}\textstyle\frac{|S|}{c(G-S)}\; :\; S\subset V\; \mbox{and}\; c(G-S)\geq 2 \big{\}},\]
where we set $\tau (G)= \infty$ if $G$ is a complete graph.
Note that $\tau(G)\ge 1$ if $G$ is hamiltonian; the reverse statement does not hold in general (see~\cite{BM08}).
The {\sc Toughness} problem is to compute $\tau(G)$ for a graph $G$.
Bauer, Hakimi and Schmeichel~\cite{BHS90} showed that already deciding whether $\tau(G)=1$ is $\coNP$-complete.

The \emph{scattering number} of a graph $G=(V,E)$ was defined by Jung~\cite{Ju78} as
$$\scat(G)=\max\{c(G-S)-|S| \; :\; S\subset V\; \mbox{and}\;  c(G-S)\geq 2\},$$ 
where we set $\scat(G)=-\infty$ if $G$ is a complete graph.
We call a set $S$ on which $\scat(G)$ is attained a \emph{scattering set}.
Note that $\scat(G)\le 0$ if $G$ is hamiltonian.
Shih, Chern and Hsu~\cite{SCH92} show  that
$\scat(G)\leq \pi(G)$ for all graphs $G$. Hence, $\scat(G)\le 1$ if $G$ is traceable.
The {\sc Scattering Number} problem is to compute $\scat(G)$ for a graph $G$.
The observation that $\scat(G)=0$ if and only if $\tau(G)=1$ combined with the aforementioned result of Bauer, Hakimi and Schmeichel~\cite{BHS90}  implies that already deciding whether 
$\scat(G)=0$ is $\coNP$-complete.

A graph $G$ is an {\it interval graph} if it is the intersection graph of a set of closed intervals on the real line, i.e., the vertices of $G$ correspond to the intervals and two vertices are adjacent in $G$ if and only if their intervals have at least one point in common. An interval graph is {\it proper} if it has a closed interval representation in which no interval is properly contained in some other interval. 

\subsection{Known Results}\label{s-known}

We first discuss the results on testing hamiltonicity properties for proper interval graphs. Besides giving a linear-time algorithm for solving {\sc Hamilton Cycle} on proper interval graphs, Bertossi~\cite{Be83} also showed that a proper interval graph is traceable if and only if it is connected~\cite{Be83}. His work was extended by Chen, Chang and Chang~\cite{CCC97} who showed that
a proper interval graph is hamiltonian if and only if it is 2-connected, and that a proper interval graph is Hamilton-connected if and only if it is 3-connected. 
In addition, Chen and Chang~\cite{CC96} showed that a proper interval graph has scattering number at most $2-k$ if and only if it is $k$-connected.

Below we survey the results on testing hamiltonicity properties for interval graphs that appeared after the aforementioned result of Keil~\cite{Ke85} on solving {\sc Hamilton Cycle} for interval graphs  in $O(m+n)$  time.

\medskip
\noindent
{\it Testing for Hamilton cycles and Hamilton paths.} 
The $O(m+n)$ time algorithm of Keil~\cite{Ke85}
makes use of an interval representation. One can find such a representation by executing the $O(m+n)$ time interval recognition algorithm of 
Booth and Leuker~\cite{BL76}. If an interval representation is already given,
Manacher, Mankus and Smith~\cite{MMS90} showed that {\sc Hamilton Cycle} and {\sc Hamilton Path} can be solved in $O(n\log n)$ time. In the same
paper, they ask whether the time bound for these two problems can be improved to $O(n)$ time if a so-called sorted interval representation is given.
Chang, Peng and Liaw~\cite{CPL99} answered this question in the affirmative. They showed that this even holds for {\sc Path Cover}.

\medskip
\noindent
{\it When no Hamilton path exists.}
In this case, {\sc Longest Path} and {\sc Path Cover} are natural problems to consider.
Ioannidou, Mertzios and Nikolopoulos~\cite{IMN11} gave an $O(n^4)$ algorithm for solving {\sc Longest Path} on interval graphs.
Arikati and Pandu Rangan~\cite{AP90} and also Damaschke~\cite{Da93}
showed that {\sc Path Cover} can be solved in $O(m+n)$ time on interval graphs. 
Damaschke~\cite{Da93}  
posed the complexity status of {\sc 1-Hamilton Path} and {\sc 2-Hamilton Path} on interval graphs as open questions. 
The latter question is still open, but Asdre and Nikolopolous~\cite{AN10} answered the former 
question  by presenting an $O(n^3)$ time algorithm that solves {\sc 1-Path Cover}, and hence {\sc 1-Hamilton Path}.
Li and Wu~\cite{LW} announced an $O(m+n)$ time algorithm for {\sc 1-Path Cover} on interval graphs.
Although Hung and Chang~\cite{HC11} do not mention the scattering number explicitly, they show that for all $k\geq 1$
an interval graph has a path cover of size at most $k$ if and only if its scattering number is at 
most $k$. Moreover, they give an $O(n+m)$ time algorithm that finds a scattering set of an interval graph $G$ with $\scat(G)\geq 0$.
They also prove that an interval graph $G$ is hamiltonian if and only if $\scat(G)\leq 0$. Recall that the latter condition is equivalent to $\tau(G)\geq 1$. 
As such, their second result is claimed~\cite{CJKL98,KKS07} to be implicit already in Keil's algorithm~\cite{Ke85}.

  \subsection{Our Results}\label{s-ours}
 
 {\it When a Hamilton path does exist.} 
 In this case, {\sc Hamilton Connectivity} is a natural problem to consider.
 Isaak~\cite{Is98} used a closely related variant of toughness called $k$-path toughness to characterize interval graphs that contain the $k$th power of a Hamiltonian path. However, the aforementioned results of Hung and Chang~\cite{HC11} suggest that trying to characterize $k$-Hamilton-connectivity in terms
 of the scattering number of an interval graph may be more appropriate than doing this in terms of its toughness. We confirm this by showing that for all $k\geq 0$ an interval graph is $k$-Hamilton-connected if and only if its scattering number is at most $-(k+1)$. Together with the results of Hung and Chang~\cite{HC11} this leads to the following theorem.
  
\begin{theorem}\label{t-char}
Let $G$ be an interval graph. Then $\scat(G)\leq k$ if and only if\\[-15pt]
\begin{enumerate}[\quad(i)]
\item $G$ has a path cover of size at most $k$ when $k\geq 1$
\item $G$ has a Hamilton cycle when $k=0$
\item $G$ is $-(k+1)$-Hamilton-connected when $k\leq -1$.
\end{enumerate}
\end{theorem}
Moreover, we give an $O(m+n)$ time algorithm for solving {\sc Scattering Number} that also produces a scattering set. This improves the $O(n^4)$ time bound of a previous algorithm due to Kratsch, Kloks and M\"uller~\cite{KKM94}. Combining this result with Theorem~\ref{t-char} yields that {\sc Hamilton Connectivity} can be solved in $O(m+n)$ time on interval graphs.
For proper interval graphs we can express $k$-Hamilton-connectivity also in the following way.
 Recall that a proper interval graph has scattering number at most $2-k$ if and only if it is $k$-connected~\cite{CC96}.
Combining this result with Theorem~\ref{t-char} yields that for all $k\geq 0$, a proper interval graph is $k$-Hamilton-connected if and only if it is $(k+3)$-connected.

\subsection{Our Proof Method}\label{s-method}

In order to explain our approach we first need to introduce some additional terminology.
A set of internally vertex-disjoint paths $P_1,\ldots,P_p$, all of which have the same  
end-vertices $u$ and $v$ of a graph $G$, is called a {\it stave} or {\it $p$-stave} of $G$, which is {\it spanning} if $V(P_1)\cup \cdots \cup V(P_p)=V(G)$.
A spanning $p$-stave between two vertices $u$ and $v$ is also called  
a spanning $(p;u,v)$-path-system \cite{CCLLW}, a $p^*$-container between $u$ and $v$~\cite{Hs94,LHH07} or a spanning $p$-trail~\cite{LW}.
 By Menger's Theorem (Theorem 9.1 in \cite{BM08}), a graph $G$ is $p$-connected if and only if there exists a $p$-stave between any pair of vertices of $G$. It is also well-known that the existence of a $p$-stave between two given vertices can be decided in polynomial time (cf. \cite{BM08}). However, given an integer $p\ge 1$ and two vertices $u$ and $v$ of a general input graph $G$, deciding whether there exists a spanning $p$-stave between $u$ and $v$ is clearly an $\NP$-complete problem:  for $p=1$ there is a trivial polynomial reduction from the $\NP$-complete problem of deciding whether a graph is Hamilton-connected; for $p=2$ the problem is equivalent to the $\NP$-complete problem of deciding whether a graph is hamiltonian; for $p\ge 3$, the $\NP$-completeness follows easily by induction and by considering the graph obtained after adding one vertex and joining it by an edge to $u$ and $v$. 
We call a spanning stave between two vertices $u$ and $v$ of a graph {\em optimal\/} if it is a $p$-stave and there does not exist a spanning $(p+1)$-stave between $u$ and $v$.

Damaschke's algorithm~\cite{Da93}  for solving {\sc Path Cover} on interval graph, which is based on the approach of Keil~\cite{Ke85}, actually solves the following problem in $O(m+n)$ time: given an interval graph $G$ and an integer $p$, does $G$ have a spanning $p$-stave between the vertex $u_1$  corresponding to the leftmost interval of an interval model of $G$ and the vertex $u_n$ corresponding to the rightmost one?  We extend Damaschke's algorithm in Section~\ref{s-algo} to an $O(m+n)$ time algorithm that takes as input only an interval graph $G$ and finds an optimal stave of $G$ between $u_1$ and $u_n$, unless it detects that there does not exist a spanning stave between $u_1$ and $u_n$. 
In the latter case  $G$ is not hamiltonian. Hence,  $\scat(G)\geq 1$ as shown by Hung and Chang~\cite{HC11} meaning that their $O(m+n)$ time algorithm for computing a scattering set may be applied. 
Otherwise, i.e., if our algorithm found an optimal stave between $u_1$ and $u_n$, we show how this enables us to compute a scattering set of $G$ in $O(m+n)$ time. In the same section, we derive
that $G$ contains a spanning $p$-stave between $u_1$ and $u_n$ if and only if $\scat(G)\le 2-p$.

In Section~\ref{s-char} we prove our contribution to Theorem~\ref{t-char} (iii), i.e., the case when $k\leq -1$. In particular, for proving the subcase $k=-1$, we show that an interval graph $G$  is Hamilton-connected if it contains a spanning $3$-stave between 
the vertex  corresponding to the leftmost interval of an interval model of $G$ and the vertex corresponding to the rightmost one.

\section{Spanning Staves and the Scattering Number}\label{s-algo}

In order to present our algorithm we start by giving the necessary terminology and notations.

A set $D\subseteq V$ \emph{dominates} a graph $G=(V,E)$ if each vertex of $G$ belongs to $D$ or has a neighbor in $D$. 
We will usually denote a path in a graph by its sequence of distinct vertices such that consecutive vertices are adjacent.
If $P=u_1\dots u_n$ is a path, then we denote its \emph{reverse} by $P^{-1}=u_n\dots u_1$.  
We may concatenate two paths $P$ and $P'$ whenever they are vertex-disjoint except for the last vertex of $P$ coinciding with the first vertex of $P'$.
The resulting path is then denoted by $P\circ P'$.

A {\it clique path} of an interval graph $G$ with vertices $u_1,\ldots,u_n$ is a sequence $C_1,\ldots,C_s$ of all maximal cliques of $G$, such that each edge of $G$ is present in some clique $C_i$ and
each vertex of $G$ appears in consecutive cliques only. This yields a specific interval model for $G$ that we will use throughout the remainder of this paper: a vertex $u_i$ of $G$ is represented by the interval 
$I_{u_i}=[\ell_i,r_i]$, where $\ell_i=\min \{j : u_i\in C_j\}$ and $r_i=\max \{j : u_i\in C_j\}$, which are referred to as the {\it start point} and the {\it end point} of $u_i$, respectively.
By definition, $C_1$ and $C_s$ are maximal cliques. Hence both $C_1$ and $C_s$ contain at least one vertex that does not occur in any other clique.
We assume that $u_1$ is such a vertex in $C_1$ and that $u_n$ is such a vertex in $C_s$. Note that $I_{u_1}=[1,1]$ and $I_{u_n}=[s,s]$ are single points. 

Damaschke made the useful observation that any Hamilton path in an interval graph can be reordered into a monotone one, in the following sense.

\begin{lemma}[\cite{Da93}]\label{lem:eat}
If the interval graph $G$ contains a Hamilton path, then it contains a Hamilton
path from $u_1$ to $u_n$. 
\end{lemma}

We use Lemma~\ref{lem:eat} to rearrange certain path systems in $G$ into a single path as follows. Let $P$ be a path between $u_1$ and $u_n$ and let $\cQ=(Q_1,\dots,Q_k)$ be a collection of paths, each of which contains $u_1$ or $u_n$ as an end-vertex. Furthermore, $P$ and all the paths of $\cQ$ are assumed to be vertex-disjoint except for possible intersections at $u_1$ or $u_n$. Consider the path $Q_1$. By symmetry, it may be assumed to contain $u_1$. We
apply Lemma~\ref{lem:eat} to $P\circ (Q_1-u_n)$ and obtain a path $P'$ between $u_1$ and $u_n$
containing all the vertices of $P\cup Q_1$. Proceeding in a similar way for the paths $Q_2,\dots,Q_k$, we obtain a path between $u_1$ and $u_n$
on the same vertex set as $P\cup\bigcup_{j=1}^k Q_j$. We denote the resulting path by
$\merg(P,Q_1,\dots,Q_k)$ or simply by $\merg(P,\cQ)$.

Let $G$ be an interval graph with all the notation as introduced above. In particular, the vertices of $G$ are $u_1,\ldots, u_n$, we consider a clique path $C_1,\dots,C_s$, and the start point and end point of each $u_i$ are $\ell_i=\min \{j : u_i\in C_j\}$
and $r_i=\max \{j : u_i\in C_j\}$, respectively, 
where $I_{u_1}=[1,1]$ and $I_{u_n}=[s,s]$.
We can obtain this representation of $G$ by first executing the $O(m+n)$ time recognition algorithm of interval graphs due to Booth and Lueker~\cite{BL76} as their algorithm also produces a clique path $C_1,\dots,C_s$ for input interval graphs.

Algorithm~\ref{alg:maxkstave} is our $O(m+n)$ time algorithm for finding an optimal stave between $u_1$ and $u_n$ if it exists. 
It gradually builds up a set $\cP$ of internally disjoint paths starting at $u_1$ and passing through vertices of $C_t\setminus C_{t+1}$ before moving to $C_t\cap C_{t+1}$ for $t=1,\ldots,s$. It is convenient to consider all these paths ordered from $u_1$ to their (temporary) end-vertices that we call {\em terminals\/}, and to use the terms
{\em predecessor}, {\em successor}, and {\em descendant\/} of a fixed vertex $v$ in one of the paths with the usual meaning of a vertex immediately before, immediately after, and somewhere after $v$ in one of these paths, respectively. 

\begin{algorithm}[h]
\KwIn{A clique-path $C_1,\dots,C_s$ in an interval graph $G$.}
\KwOut{An optimal spanning stave $\cP$ between $u_1$ and $u_n$, if it exists.}
\Begin{
let $p=\deg(u_1)$\;
let $R_i=u_1$ for all $i=1,\dots,p$\;
let $\cP=\{R_1,\dots,R_p\}$\;
let $\cQ=\emptyset$\;
\For{$t:=1$ to $s-1$}{
choose a $P\in \cP$ whose terminal has the smallest end point among all terminals\;
\lIf {$C_t\setminus(C_{t+1}\cup\bigcup (\cP\cup \cQ))\ne \emptyset$}
{extend $P$ by $C_t\setminus(C_{t+1}\cup\bigcup (\cP\cup \cQ))$}\;
\For {every path $R\in \cP$}
{\If {the terminal of $R$ is not in $C_{t+1}$}{
{try to extend each $R$ by a new vertex $u$ from $(C_t\cap C_{t+1})\setminus \bigcup (\cP\cup \cQ)$ with the smallest end point\; 
\If {such $u$ does not exist} 
{remove $R$ from $\cP$\;
insert $R$ into $\cQ$\;
decrement $p$ \;
\lIf {$p=0$} {report that $G$ has no spanning 1-stave between $u_1$ and $u_n$ and \bf quit} 
}}}}}
choose any $P\in \cP$\;
extend $P$ by all vertices of $C_s\setminus \bigcup (\cP\cup \cQ)$\;
{let $P=\merg(P,\cQ)$}\;
\lFor {every path $R\in \cP\setminus P$}
{extend $R$ by $u_n$}\;
report an optimal spanning $p$-stave $\cP$;
}
\caption{Finding an optimal spanning stave.}\label{alg:maxkstave}
\end{algorithm}

Before we prove the correctness of Algorithm~\ref{alg:maxkstave}, we develop some more auxiliary terminology related to this algorithm. 

We say that a vertex $v$ has been {\em added to a path\/} if, at some point in the execution of Algorithm~\ref{alg:maxkstave}, some path $R\in\cP$ such that $v\notin V(R)$ has been extended to a longer path containing $v$ (and possibly some other new vertices). If $u_i$ has been processed by the algorithm and added to a path at lines 8 or 
11 of Algorithm~\ref{alg:maxkstave}, we 
say that $u_i$ has been \emph{activated} at time $a_i$, and we assign $a_i$ the current value of the variable $t$. Thus, we think of time steps $t=1,\ldots,t=s$ during the execution of the algorithm. 
When 
at the same or a later stage a vertex $u_j$ has been added as a successor of $u_i$ to a path,
we say that $u_i$ has been \emph{deactivated} at time $d_i$, and assign $d_i=a_j$. 
Hence, as soon as $a_i$ and $d_i$ have assigned values, we have $\ell_i\le a_i\le d_i\le r_i$. Furthermore, any of the implied inequalities holds whenever both of its sides are defined.
Note that any of these inequalities may be an equality; in particular, a vertex can be activated and deactivated at the same time.

If the involved parameters have assigned values, we consider
the open (time) intervals $(\ell_i,a_i)$, $(a_i,d_i)$ and $(d_i,r_i)$, and we
say that $u_i$ is \emph{free} during $(\ell_i,a_i)$ if this interval is nonempty, \emph{active} during $(a_i,d_i)$ if this interval is nonempty, and \emph{depleted} during $(d_i,r_i)$ if this interval is nonempty.
In particular, note that the vertices that are added to a path at line 8 (if any) are from
$C_t\setminus C_{t+1}$, so they satisfy $r_i=t$ and $a_i=t$. Such vertices will not be
active or depleted during any (nonempty) time interval, but they are free during the time interval $(\ell_i,r_i)$ if this interval is nonempty.

For $1\leq j \leq k \leq s$, we define
\begin{equation*}
  C_{j,k} = (\bigcup_{i=j}^kC_i).
\end{equation*}

The following lemma is crucial.

\begin{lemma}\label{lem:induction}
Suppose that Algorithm~\ref{alg:maxkstave} terminates at line 
16 or finishes an iteration of the loop at lines 6--20. Let the current value of the variable $t$ be also denoted by $t$. If there is at least one depleted vertex during the interval $(t,t+1)$, then there exists an integer $t'<t$ with the following properties (see Fig.~\ref{fig:induction}a for an illustration):
\begin{enumerate}[\quad(i)]
\item $C_{t'+1,t}\setminus (C_{t'} \cup C_{t+1})\neq\emptyset$,
\item a unique vertex $u_i\in C_{t'}\cap C_{t+1}$ is active during $(t',t'+1)$ and is depleted during $(t,t+1)$, 
\item all vertices that are active during $(t,t+1)$ are also active during $(t',t'+1)$, with the only possible exception of 
the last descendant of $u_i$ (which we denote by $v$) 
that can be free during $(t',t'+1)$, 
\item all vertices that are depleted during $(t,t+1)$ and distinct from $u_i$ are also depleted during $(t',t'+1)$,
\item all vertices that are active during $(t',t'+1)$ are also active during $(t,t+1)$, with the only exception of $u_i$, and 
\item all vertices that are free during $(t',t'+1)$ are also free during $(t,t+1)$, with the only possible exception of $v$ if it is active during $(t,t+1)$.
\end{enumerate}
\end{lemma}

\begin{figure}[h]
\begin{center}
\includegraphics{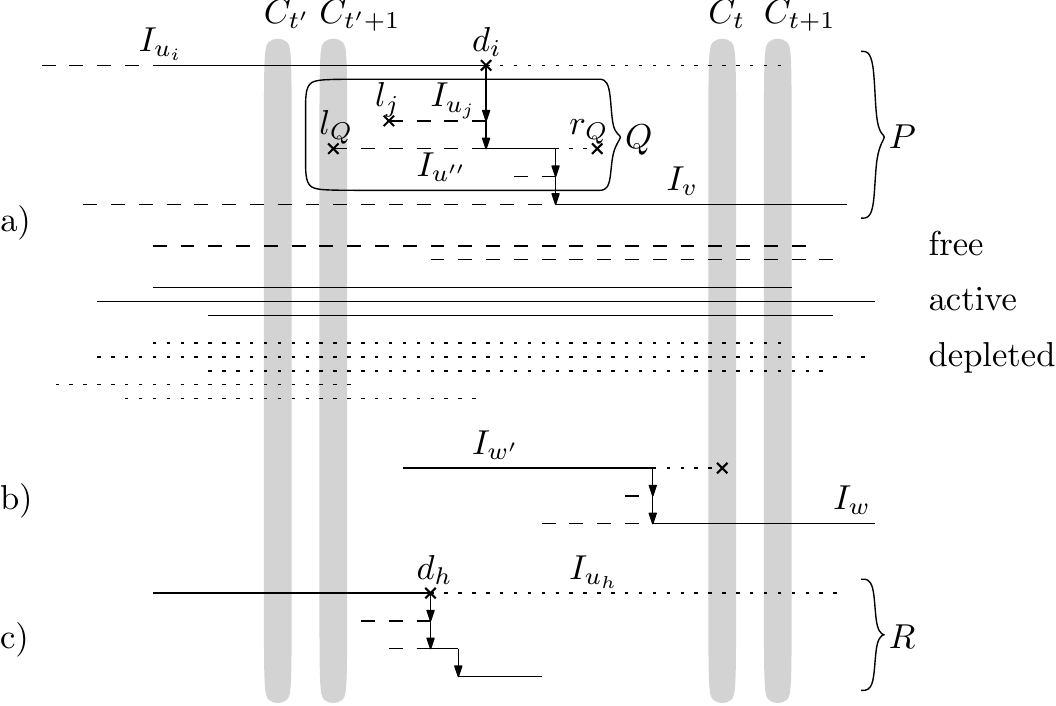}
\end{center}
\caption{A path system as described in Lemma~\ref{lem:induction}. The vertical arrows indicate successors in the paths and the time of activation and deactivation.}\label{fig:induction}
\end{figure}

\begin{proof}
Assume that there is at least one depleted vertex during the interval $(t,t+1)$, and
let $u_i$ be a vertex with the latest deactivation time among those that are depleted during $(t,t+1)$.
To prove that this vertex is unique, we note that all but at most one of the vertices deactivated during a given iteration of the loop on lines 
6--20 (say, at time $t$) have end point equal to $t$ and hence cannot be depleted during a nonempty interval. The only possible exception is the terminal of the path $P$ chosen at line 7 (and only if it is deactivated due to adding a vertex to $P$ at line 8).

We define $Q$ to be the subpath of $P$ formed by all descendants of $u_i$, except that
if the last descendant $v$ of $u_i$ is active during $(t,t+1)$, we do not include $v$ in $Q$.
Observe that the successor of $u_i$ has the same deactivation time as $u_i$, 
hence it is distinct from $v$, and therefore  $Q$ is nonempty.
Let $\ell_Q$ be the smallest start point among intervals corresponding to vertices of $Q$, and let $r_Q$ be the largest such end point. 

If $P$ has a vertex that is active during $(t,t+1)$, this vertex is $v$ and it is not a vertex of $Q$.
Thus all vertices of $Q$ are either depleted during $(t,t+1)$ or their end point is less than or equal to $t$. 
By the choice of $u_i$, none of them belongs to $C_{t+1}$, and hence $r_Q\le t$.
We choose $t'=\ell_Q-1$. Notice that for $u_j\in V(Q)$, $r_j \geq d_i$. Thus if we let $u_q$ be the vertex of $u$ such that $\ell_q = \ell_Q$, then $u_q$ is free during $(t'+1,d_i)$.

Clearly, all vertices of $Q$ are in $C_{t'+1,t}\setminus (C_{t'} \cup C_{t+1})$. Hence, this set is not empty and property (i) is proved.

We prove (ii). Since the deactivation of $u_i$ happened when its successor $u_j$ was free, we have
$d_i\ge \ell_j > t'$. Hence, $u_i$ cannot be depleted during $(t',t'+1)$.
Clearly, $u_i\neq u_1$, as $u_1$ is not depleted during $(t-1,t)$. Therefore,
$u_i$ has a predecessor. Denote it by $u'$.
If $u'$ were adjacent to the vertex $u_q$ of $Q$, then the algorithm would choose $u_q$ as the successor of $u'$, since $r_i > r_Q \geq r_q$. Consequently, the start point of $u'$ is less than or equal to $t'$, so $u_i$ is active during $(t',t'+1)$. The uniqueness of $u_i$ will follow easily once we establish property (iv).

To show property (iii), assume that $u_m$ is a vertex different from $v$ that is active during $(t,t+1)$
but has been activated after $t'$. 
Since $u_1$ is not active during $(t,t+1)$, $u_m\neq u_1$ and $u_m$ has a predecessor $u'$.
We first suppose that $u_m$ is active during $(d_i-1,d_i)$. 
The vertex $u'$ is deactivated at some time $t''$ such that $t'+1\leq t''\leq d_i-1$. Hence, it is adjacent to the previously defined vertex $u_q$ of $Q$ that is free during $(t'+1,d_i)$. Since $r_q \leq r_Q < t+1 \leq r_m$, the successor of $u'$ should be $u_q$ rather than $u_m$, a contradiction.   

It follows that $u_m$ is not active during $(d_i-1,d_i)$. 
The vertex $u_m$ is included in some path $R\in\cP$, $R\neq P$. This path contains a vertex $w'$ that is active 
during $(d_i-1,d_i)$ (see Fig.~\ref{fig:induction}b), where $u_m$ is a descendant of $w'$.
Observe that $w'$ is not active during $(t,t+1)$ because $u_m$ is.
Suppose that the end point of $w'$ is at least $t+1$.  Then $w'$ is depleted during $(t,t+1)$, so by the choice of $u_i$, $w'$ is deactivated before time $d_i$ and cannot be active during $(d_i-1,d_i)$, a contradiction.

Thus, the end point of $w'$ is not larger than $t$.  But then $w'$
should have been chosen at line 7 of the algorithm instead of $u_i$.

For (iv), assume that some $u_h\ne u_i$ is depleted during $(t,t+1)$, but $d_h\geq t'+1$. 
By the choice of $u_i$, we have $d_h<d_i$. 
Without loss of generality, assume that $u_h$ was chosen such that $d_h$ is maximal.
Let $R$ be the path in $\cP\cup\cQ$ containing $u_h$. Note that $R\neq
P$. If $R$ contains a vertex $w$ that is active during $(t,t+1)$, then by (iii),
$w$ is active during $(t',t'+1)$ and we conclude that $u_h$ cannot be included in $R$; a contradiction.

It follows that no vertex of $R$ is active during $(t,t+1)$ (see Fig.~\ref{fig:induction}c).
Moreover, by the choice of $u_h$, the end points of all its descendants are less than or equal to $t$, because if there is a descendant $u_j$ of $u_h$ with $r_j \geq t+1$, then $w$ is depleted during $(t,t+1)$ and $d_j > d_h$, a contradiction.   
Recall that the vertex $u_q$ is free during $(t'+1,d_i)$. Since the path $R$ cannot be terminated while a free vertex is available, it must contain a vertex that is active during $(d_i-1,d_i)$. However, this vertex has a smaller end point than $u_i$, contradicting the correct execution of the algorithm at line 7.

To obtain (v), assume that $w\neq u_i$ is active during $(t',t'+1)$ but not active during $(t,t+1)$. The vertex $w$ is included in some path $R\in\cP\cup\cQ$, $R\neq P$. 
If  one of the descendants of $w$ is active during $(t,t+1)$, then by (iii), this vertex is active during $(t',t'+1)$ contradicting the activeness of $w$ at the same time.
Similarly, if $w$ or one of its descendants is depleted during $(t,t+1)$, then by (iv), this vertex is depleted during $(t',t'+1)$ and $w$ cannot be active.
It follows that the end points of $w$ and its descendants are less than or equal to $t$. If $d_i=t'+1$, then $R$ has a vertex that is active during  $(d_i-1,d_i)$.
If $d_i>t'+1$, then we use the observation that the vertex $u_q$ is free during $(t'+1,d_i)$, and again conclude that
 $R$ has an active vertex during  $(d_i-1,d_i)$. Then this vertex should be selected by the algorithm in line 7 instead of $u_i$; a contradiction.

It remains to prove (vi). Let $w$ be a vertex that is free during $(t',t'+1)$ and not free during $(t,t+1)$. Moreover, we assume that $w\neq v$ if $v$ is active during $(t,t+1)$.
Our algorithm does not terminate until time $t$. Therefore, $w$ is included in some path $R\in\cP\cup\cQ$, $R\neq P$. This path has  a vertex that is active during $(t',t'+1)$. 
By (v), this vertex remains active until $t+1$, but it means that $w$ is not included in $R$. \qed
\end{proof}

Now we are ready to state and prove the main structural result.

\begin{theorem}\label{thm:syst}
An interval graph $G$ contains a spanning $p$-stave between $u_1$ and $u_n$ 
if and only if $\scat(G)\le 2-p$.
\end{theorem}

\begin{proof}
Let us first assume that $\cP=(R_1\dots,R_p)$ is a spanning $p$-stave between $u_1$ and $u_n$. If $G$ is complete, then the claim is trivial. Otherwise, let $S\subset V(G)$ be a scattering set. We claim that $u_1,u_n\notin S$. Suppose the contrary. Since the neighborhood of $u_1$ induces a clique, $c(G-S) \leq c(G-(S-\{u_1\}))$ and therefore 
\begin{equation*}
c(G-S)-|S| < c(G-(S-\{u_1\})) - |S-\{u_1\}|,
\end{equation*}
a contradiction with the choice of $S$. The argument for $u_n$ is symmetric.

The internal vertices of each path in $\cP$ dominate $G$. Hence, 
the vertex cut $S$ contains an internal vertex from each path of $\cP$. From each path $R_i$ of $\cP$, we choose a vertex $s_i \in S$ and set $S' = \{s_1,\dots,s_p\}$. 

Consider the spanning subgraph $G'$ of $G$ induced by the edges of $\cP$.
Observe that $G'-S'$ has two components. If we remove the remaining vertices of
$S\setminus S'$ one by one, then with each vertex we remove, 
the number of components of the remaining graph can increase by at most one as $u_1,u_n\notin S$.
Hence $c(G-S)\le c(G'-S)\le 2+|S|-p$ and 
$\scat(G)\le 2-p$, proving
the forward implication of the statement.
 
For the other direction, let us assume that $G$ does not have a spanning $p$-stave between $u_1$ and $u_n$. During the execution of Algorithm~\ref{alg:maxkstave}, at some stage the value set at line 14 becomes smaller than $p$. 
Suppose $t_1$ is the value of the variable $t$ at this moment.
We will complete the proof by constructing a scattering set $S$ and showing that for this set $c(G-S)-|S|> 2-p$.

We repeatedly use Lemma~\ref{lem:induction} and find a finite sequence 
$t_1,t_2,\dots,t_k$, such that $t_{i+1}=(t_i)'$ as long as there are depleted vertices during $(t_i,t_i+1)$ for $i<k$.
Notice that there are no depleted vertices during $(1,2)$, i.e., this process stops and we have no depleted vertices during $(t_k,t_k+1)$.
We choose $S=\bigcup_{i=1}^k (C_{t_i}\cap C_{t_{i+1}})$ and prove  
that $G-S$ has at least $|S|-p+3$ components.

The subgraphs $G[C_{1,t_k}]-S$ and $G[C_{t_1+1,s}]-S$ contain $u_1$ and $u_n$, respectively; in particular, they have at least one component each. By property (i) in Lemma~\ref{lem:induction}, $G[C_{t_{i+1}+1,t_i}]-S$ has at least one component for each $i\in\{1,\ldots,k-1\}$. 
Since all these components are distinct components of $G-S$, the graph $G-S$ has at least $k+1$ components.

By properties (ii), (v) and (vi) in Lemma~\ref{lem:induction}, $(C_{t_{i+1}}\cap C_{t_{i+1}+1})\setminus (C_{t_{i}}\cap C_{t_{i}+1}) $ contains only vertices that are depleted during $(t_{i+1},t_{i+1}+1)$ for each $i\in\{1,\ldots,k-1\}$. Further, $C_{t_1}\cap C_{t_1+1}$  has no vertices that are free during $(t,t+1)$, because at least one path is not extendable at time $t_1$. Also 
this set has at most  $p-1$ vertices that are active during $(t,t+1)$. Hence, the remaining vertices are depleted.
By properties (ii) and (iv) in Lemma~\ref{lem:induction}, for each $i\in\{1,\ldots,k-1\}$, exactly one vertex that is depleted during $(t_i,t_{i+1})$ has a different status during $(t_{i+1},t_{i+1}+1)$ and is active.
It follows that 
\begin{equation*}
  |S| \leq (p-1) + (k-1) = k+p-2
\end{equation*}
as required.\qed
\end{proof}

Recall that 
the scattering number can be determined in $O(m+n)$ time by an algorithm of Hung and Chang~\cite{HC11}
if the scattering number is positive. 
Then, by analyzing Algorithm~\ref{alg:maxkstave}, we get the following result:

\begin{corollary}\label{c-scat}
The scattering number as well as a scattering set of an interval graph can be computed in $O(m+n)$ time.
\end{corollary}
The only operation whose time complexity has not been discussed is $\merg(P,\cQ)$ at line 21. 
We refer to Damaschke's proof of Lemma~\ref{lem:eat} to verify that this can be implemented in $O(m+n)$ time.

Our proof of Theorem~\ref{thm:syst} provides a construction 
of a scattering set that can be straightforwardly implemented in linear time.

\section{Hamilton-connectivity}\label{s-char}

In this section we prove our contribution to  Theorem~\ref{t-char}, which is the following.

\begin{theorem}\label{t-ham}
For all $k\geq 0$, an interval graph $G$ is $k$-Hamilton-connected if and only if  $\scat(G)\leq -(k+1)$.
\end{theorem}

\begin{proof}
Let  $k\geq 0$ and $G$ be an interval graph
with leftmost and rightmost vertices $u_1$ and $u_n$ as defined before. The statement of Theorem~\ref{t-ham} is readily seen to hold when $G$ is a complete graph.
Hence we may assume without loss of generality that $G$ is not complete.

First suppose that $G$ is $k$-Hamilton-connected. Then $G$ has at least $k+3$ vertices.
We claim that $G-R$ is traceable for every subset $R\subset V(G)$ with $|R|\leq k+2$.
In order to see this, suppose that $R\subseteq V(G)$ with $|R|\leq k+2$.
We may assume without loss of generality that $|R|=k+2$.
Let $s$ and $t$ be two vertices of $R$. By definition, $G^*=G-(R\setminus \{s,t\})$ has a Hamilton path with end-vertices $s$ and $t$. 
Hence $G-R=G^*-\{s,t\}$ is traceable. Below we apply this claim twice.

Because $G$ is not complete, $G$ has a scattering set $S$. By definition, $S$ is a vertex cut. Hence $S=\{s_1,\ldots,s_\ell\}$ for some $\ell\geq k+3$, as otherwise
$G-S$ would be traceable, and thus connected, due to our claim. 
Let $T=\{s_1,\ldots,s_{k+2}\}$ and let $U=\{s_{k+3},\ldots,s_\ell\}$.
By our claim, $G'=G-T$ is traceable implying that $\scat(G')\leq 1$~\cite{SCH92}.
Because $c(G'-U)=c(G-S)\geq 2$, we find that $U$ is a vertex cut of $G'$. We use these two facts to derive that
\[
\begin{array}{lcl}
1 &\geq &\scat(G')\\[3pt]
   &\geq &c(G'-U)-|U|\\[3pt]
   &=&c(G-T-U)-|T|-|U|+|T|\\[3pt]
   &= &c(G-S)-|S|+|T|\\[3pt]
   &= &\scat(G)+|T|\\[3pt]
   &=&\scat(G)+k+2,
\end{array}
\]
implying that $\scat(G)\leq 1-(k+2)=-(k+1)$, as required.

\begin{figure}
\begin{center}
\includegraphics{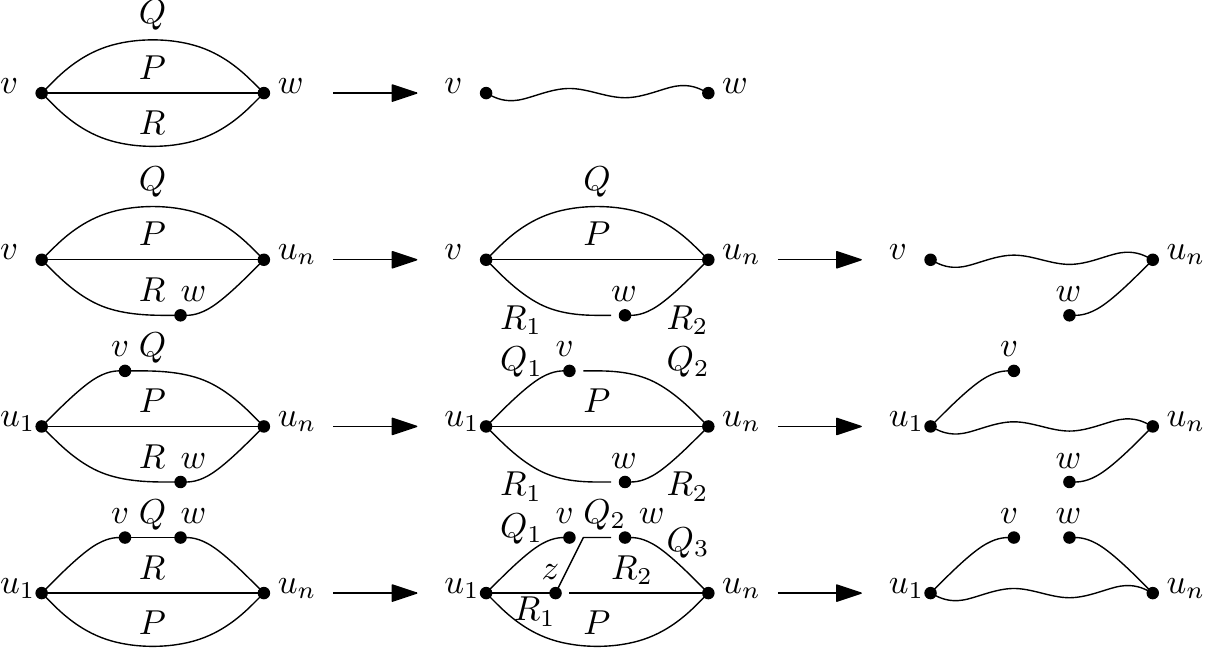}
\end{center}
\caption{The essential cases in the proof of Theorem~\ref{t-ham} for $k=0$.}\label{fig:3hamconn}
\end{figure}

Now suppose that $\scat(G)\leq -(k+1)$. 
First let $k=0$.
By Theorem~\ref{thm:syst}, there exists a spanning 3-stave $\cP=(P,Q,R)$ between $u_1$ and $u_n$. 
Let $v,w$ be an arbitrary pair of vertices  of $G$. We distinguish four cases in order to find a Hamilton path between $v$ and $w$; see Fig.~\ref{fig:3hamconn} for an illustration.

\pcase{Case 1:} $v=u_1$ and $w=u_n$. In this case, $\merg(P,Q,R)$ is the desired Hamilton path.

\pcase{Case 2:} $v=u_1$ and $w\ne u_n$. Assume without loss of generality that $w\in R$.
We split $R$ before $w$ into the subpaths $R_1$ and $R_2$, i.e., $w$ becomes the first vertex of $R_2$ 
and it does not belong to $R_1$.
Then $\merg(P,Q,R_1)\circ R_2^{-1}$ is the desired path. The case with $v\ne u_1$ and $w=u_n$ is symmetric.

\pcase{Case 3:} $v\ne u_1$ and $w\ne u_n$ belong to different paths, say $v\in Q$ and $w\in R$. We split 
$Q$ after $v$ into $Q_1$ and $Q_2$, and we also split $R$ before $w$, as above.
Then $Q_1^{-1}\circ \merg(P,Q_2,R_1)\circ R_2^{-1}$ is the desired path.

\pcase{Case 4:} $v\ne u_1$ and $w\ne u_n$ belong to the same path, say $Q$. Without loss of generality, assume that both $v\ne u_1$ and $w\ne u_n$ appear in this order on $Q$. 
We split $Q$ after $v$ and before $w$ 
into three subpaths $Q_1,Q_2,Q_3$. If $v$ and $w$ are consecutive on $Q$, i.e., when $Q_2$ is empty, then 
$Q_1^{-1}\circ \merg(P,R)\circ Q_3^{-1}$ is the desired path.
Otherwise, let $z$ be any vertex on $R$ that is a neighbor of the first vertex of $Q_2$. Such $z$ exists since
the path $R$ dominates $G$. 
We split $R$ after $z$ into $R_1$ and $R_2$. By the choice of $z$, $R_1$ and $Q_2$ can be combined through $z$ into a valid path $R'$ containing exactly the same vertices as $R_1$ and $Q_2$ and starting at $u_1$.   
Then we choose $Q_1^{-1}\circ \merg(P,R',R_2)\circ Q_3^{-1}$.

\medskip
\noindent
Now let $k\geq 1$.
Let $S$ be a set of vertices with $|S|\leq k$. We need to show that $G-S$ is Hamilton-connected. Let $T$ be a scattering set of $G-S$ and let $S^*=S\cup T$.
Because $T$ is a scattering set of  $G-S$, we find that $S^*$ is a vertex cut of $G$. We use this to derive that
\[
\begin{array}{lcl}
\scat(G-S)&=&c(G-S-T)-|T|\\[3pt]
&= &c(G-S^*)-|S^*|+|S^*|-|T|\\[3pt]
&\leq &\scat(G)+k-0\\[3pt]
&\leq &-1.
\end{array}
\]
Then, by returning to the case $k=0$ with $G-S$ instead of $G$, we find that $G-S$ is Hamilton-connected, as required. This completes the proof of Theorem~\ref{t-ham}.\qed
\end{proof}

\section{Future Work}\label{s-con}

We conclude our paper by posing a number of open problems. We start with recalling two open problems posed in the literature.

First of all, Damaschke's question~\cite{Da93} on the complexity status of 2-{\sc Hamilton Path} is still open. Our results imply 
that we may restrict ourselves to interval graphs with scattering number equal to zero or one. This can be seen as follows. Let $G$ be an interval graph that together with two of its vertices $u$ and $v$ forms an instance
of 2-{\sc Hamilton Path}. We apply Corollary~\ref{c-scat} to compute $\scat(G)$ in $O(m+n)$ time. If $\scat(G)<0$, then $G$ is Hamilton-connected by Theorem~\ref{t-char}. Then, by definition, there exists a Hamilton path between
$u$ and $v$. If $\scat(G)>1$, then $G$ is not hamiltonian, also due to Theorem~\ref{t-char}. Hence, there exists no Hamilton path between $u$ and $v$.

Second, Asdre and Nikolopoulos~\cite{AN10} asked about the complexity status of the $\ell$-{\sc Path Cover} problem on interval graphs.
This problem generalizes $1$-{\sc Path Cover} and is to determine the size of a smallest path cover of a graph $G$ subject to the additional condition that every vertex of a given set $T$ of size $\ell$ is 
an end-vertex of a path in the path cover. The same authors show that both {\sc $\ell$-Path Cover} and $2$-{\sc Hamilton Path} can be solved in $O(m+n)$ time on proper interval graphs~\cite{AN10a}. 

The {\sc Spanning Stave} problem is that of computing the minimum value of $p$ for which a given graph has a spanning $p$-stave.
Because a Hamilton path of a graph is a  spanning $1$-stave and {\sc Hamilton Path} is $\NP$-complete, this problem is $\NP$-hard. What is the computational complexity of
{\sc Spanning Stave} on interval graphs? The following example shows that we cannot generalize Lemma~\ref{lem:eat} and apply Algorithm~\ref{alg:maxkstave} as an attempt to solve this problem.
Take the graph with four vertices $a$, $b$, $c$, $d$ and edges $ab$, $ac$, $bc$, $bd$, $cd$. The resulting graph is interval.  
However, we only have a spanning $2$-stave between $a$ and $d$ (as their degrees are $2$) but there is a spanning $3$-stave between $b$ and $c$, namely
$\{bac,bc,bdc\}$. 

Chen et al.~\cite{CCLLW} define the {\em spanning connectivity\/} of a Hamilton-connected graph $G$ as the largest integer $q$ such that $G$ has a spanning $p$-stave between any two vertices of $G$ for all integers $1\le p \le q$. So, for instance, the complete graph on $n$ vertices has spanning connectivity $n-1$, and a graph has spanning connectivity at least $1$ if and only if it is Hamilton-connected.
By the latter statement, the corresponding optimization problem {\sc Spanning Connectivity} is $\NP$-hard.
What is the computational complexity of {\sc Spanning Connectivity} on interval graphs or even proper interval graphs?

Kratsch, Kloks and M\"uller~\cite{KKM94} gave an $O(n^4)$ time algorithm for solving {\sc Toughness} on interval graphs. Is it possible to improve this bound to linear on interval graphs just as we did for 
{\sc Scattering Number}?

Finally, can we extend our $O(m+n)$ time algorithms for {\sc Hamilton Connectivity} and {\sc Scattering Number} to superclasses of interval graphs such as circular-arc graphs and cocomparability graphs? 
The complexity status of {\sc Hamilton Connectivity} is still open for both graph classes, although {\sc Hamilton Cycle} can be solved in $O(n^2\log n)$ time on circular-arc graphs~\cite{SCH92} and in $O(n^3)$ time on cocomparability graphs~\cite{DS94}. It is known~\cite{KKM94} that {\sc Scattering Number} can be solved in $O(n^4)$ time on circular-arc graphs and in polynomial time on cocomparability graphs of bounded dimension.


\begin{thebibliography}{10}

\bibitem{AP90}
S.R. Arikati and C. Pandu Rangan, Linear algorithm for optimal path cover problem on interval graphs, Information Processing Letters 35 (1990) 149--153.

\bibitem{AN10a}
K. Asdre and S.D. Nikolopoulos, A polynomial solution to the k-fixed-endpoint path cover problem on proper interval graphs, Theor. Comput. Sci. 411 (2010) 967--975.

\bibitem{AN10}
K. Asdre and S.D. Nikolopoulos, The 1-fixed-endpoint path cover problem is polynomial on interval graphs, Algorithmica 58 (2010) 679--710.

\bibitem{BHS90}
D. Bauer, S.L. Hakimi, and E. Schmeichel,
Recognizing tough graphs is $\NP$-hard,
Discrete Applied Mathematics 28 (1990) 191--195.

\bibitem{Be83}
A.A. Bertossi, Finding hamiltonian circuits in proper interval graphs, Information Processing Letters 17 (1983) 97--101.

\bibitem{BB86}
A.A. Bertossi and M.A. Bonucelli,
Hamilton circuits in interval graph generalizations,
Information Processing Letters 23 (1986) 195-200.

\bibitem{BM08}
J.A. Bondy and U.S.R. Murty, Graph Theory, vol. 244 of Graduate Texts in Mathematics, Springer Verlag, 2008.

\bibitem{BL76}
K.S. Booth and G.S. Lueker, Testing for the consecutive ones property, interval graphs, and graph planarity using pq-tree algorithms, J. Comput. Syst. Sci. 13 (1976) 335--379.

\bibitem{CPL99}
M.-S. Chang, S.-L. Peng, and J.-L. Liaw,
Deferred-query: An efficient approach for some problems on interval graphs, Networks 34 (1999) 1--10.

\bibitem{CC96}
C. Chen and C.-C. Chang, Connected proper interval graphs and the guard problem in spiral polygons,  Combinatorics and Computer Science, Lecture Notes in Computer Science Volume 1120 (1996) 39--47.

\bibitem{CCC97}
C. Chen, C.-C. Chang, and G.J. Chang, Proper interval graphs and the guard problem, Discrete Mathematics 170 (1997) 223--230.

\bibitem{CCLLW}
Y. Chen, Z.-H. Chen, H.-J. Lai, P. Li, and E. Wei, On spanning disjoint paths in line graphs, Graph and Combinatorics, to appear.

\bibitem{CJKL98}
G. Chen, M.S. Jacobson, A.E. K\'ezdy and J. Lehel,  Tough enough chordal graphs are hamiltonian, Networks 31 (1998) 29--38.

\bibitem{Ch73}
V. Chv\'atal, Tough graphs and hamiltonian circuits, Discrete Mathematics 5 (1973) 215--228.

\bibitem{Da93}
P. Damaschke, Paths in interval graphs and circular arc graphs, Discrete Mathematics 112 (1993) 49--64.

\bibitem{De93}
A.M. Dean, The computational complexity of deciding hamiltonian-connectedness, Congr. Num. 93 (1993) 209--214.

\bibitem{DS94}
J.S. Deogun and G. Steiner, Polynomial algorithms for hamiltonian cycle in cocomparability graphs, SIAM Journal on Computing 23 (1994) 520--552.

\bibitem{GJ79}
M.R. Garey and D.S. Johnson, Computers and Intractability: A Guide to the Theory of NP-Completeness, W. H. Freeman \& Co Ltd, 1979.

\bibitem{GJT76}
M.R. Garey, D.S. Johnson, and R.E. Tarjan, The planar hamiltonian circuit problem is NP-complete, SIAM Journal on Computing 5 (1976) 704--714.

\bibitem{Hs94}
D. Hsu, On container width and length in graphs, groups, and networks, IEICE Trans. Fund, E77-A (1994) 668--680.

\bibitem{HC11}
R.-W. Hung and M.-S. Chang, Linear-time certifying algorithms for the path cover and hamiltonian cycle problems on interval graphs, Appl. Math. Lett. 24 (2011) 648--652.

\bibitem{IMN11}
K. Ioannidou, G.B. Mertzios, and S.D. Nikolopoulos,
The longest path problem has a polynomial solution on interval graphs,
Algorithmica 61 (2011) 320--341.

\bibitem{Is98}
G. Isaak, Powers of hamiltonian paths in interval graphs,
Journal of Graph Theory 28 (1998) 31--38.

\bibitem{Ju78}
H.A. Jung, On a class of posets and the corresponding comparability graphs, Journal of Combinatorial Theory, Series B 24 (1978) 125--133.
 
\bibitem{KKS07}
T. Kaiser, D. Kr\'al', and L. Stacho, Tough spiders, Journal of Graph Theory 56 (2007) 23--40.

\bibitem{Ke85}
J.M. Keil,  Finding hamiltonian circuits in interval graphs, Information Processing Letters 20 (1985) 201--206.

\bibitem{KKM94}
D. Kratsch, T. Kloks, and H. M\"uller, Computing the toughness and the scattering number for interval and other graphs, Tech. Rep. Rapports de recherche no. 2237, INRIA Rennes, 1994.

\bibitem{KRV12}
R. Ku\v{z}el, Z. Ryj{\'a}\v{c}ek, and P. Vr{\'a}na,
Thomassen's conjecture implies polynomiality of $1$-hamilton-connectedness in line graphs, Journal of Graph Theory 69 (2012) 241--250.

\bibitem{LW}
P. Li and Y. Wu, A linear time algorithm for solving the 1-fixed-endpoint path cover problem on interval graphs, draft, cited in
http://math.sjtu.edu.cn/faculty/ykwu/paths-in-interval-graphs.pdf.

\bibitem{LHH07}
C.-K. Lin, H.-M. Huanga, and L.-H. Hsu,
On the spanning connectivity of graphs,
Discrete Mathematics 307 (2007) 285--289.

\bibitem{MMS90}
G.K. Manacher, T.A. Mankus, and C.J. Smith, An optimum $\Theta(n\log n)$ algorithm for finding a canonical hamiltonian path and a canonical hamiltonian circuit in a set of intervals, Information Processing Letters 35 (1990) 205--211.

\bibitem{Mu96}
H. M\"uller, Hamiltonian circuits in chordal bipartite graphs, Discrete Mathematics 156 (1996) 291--298.

\bibitem{SCH92}
W.K. Shih, T.C. Chern, and W.L. Hsu, An $O(n^2\log n)$ time algorithm for the hamiltonian cycle problem on circular-arc graphs, SIAM Journal on Computing 21 (1992) 1026--1046.

\end{thebibliography}
\end{document}